\documentclass{article} 

\usepackage{amssymb,amsmath}
\usepackage{fancybox}
\usepackage{ascmac}
\usepackage{amsthm}
\usepackage{url}
\usepackage{fullpage}

\usepackage{hyperref}

\usepackage{comment}
\usepackage{siunitx}
\usepackage{graphicx}
\usepackage{hyperref}
\usepackage{bigints}
\usepackage{here}
\usepackage{booktabs}
\usepackage{amsmath}
\usepackage{amsthm}
\usepackage{bm}
\usepackage{cleveref}
\usepackage{multirow}

\usepackage{algorithm,algorithmic}

\newtheorem{theorem}{Theorem}
\newtheorem{lemma}{Lemma}

\theoremstyle{definition}
\newtheorem{example}{Example}
\newtheorem{remark}{Remark}

\newcommand{\mss}{\mathrm{SS}}

\newcommand{\bracket}[1]{\left(#1\right)}

\newcommand{\abs}[1]{\left|#1\right|} % absolute value | |
\newcommand{\norm}[1]{\left\|#1\right\|} % norm \| \|
\newcommand{\set}[3][undefined]{\expandafter\ifx\csname#1\endcsname\undefined \left\{#2\,\middle|\,#3\right\}%
\else\csname#1l\endcsname\{#2\,\csname#1\endcsname|\,#3\csname#1r\endcsname\}\fi}

\newcommand{\setZ}{\mathbb{Z}}

\newcommand{\setR}{\mathbb{R}}

\newcommand{\vecfont}[1]{\mathbf{#1}}
\newcommand{\matfont}[1]{\mathbf{#1}}

\usepackage{latexsym}

\DeclareMathOperator{\vol}{vol}
\DeclareMathOperator{\Pot}{Pot}
\DeclareMathOperator{\argmin}{argmin}
\DeclareMathOperator{\argmax}{argmax}
\DeclareMathOperator{\Span}{span}
\DeclareMathOperator{\id}{id}

\title{Halt Properties and Complexity Evaluations for Optimal DeepLLL Algorithm Families%
\footnote{This research was mainly done when the first author was an undergraduate student at Department of Mathematical Engineering and Information Physics, School of Engineering, The University of Tokyo and when the second author was with Graduate School of Information Science and Technology, The University of Tokyo.}}

\author{Takuto Odagawa${}^1$ \and Koji Nuida${}^2{}^3$}
\date{%
${}^1$ Graduate School of Information Science and Technology, The University of Tokyo, Japan \\
\url{tk55-1234@g.ecc.u-tokyo.ac.jp} \\
${}^2$ Institute of Mathematics for Industry (IMI), Kyushu University, Japan \\
\url{nuida@imi.kyushu-u.ac.jp} \\
${}^3$ National Institute of Advanced Industrial Science and Technology (AIST), Japan%
}

\begin{document}

\maketitle

\begin{abstract}
DeepLLL algorithm (Schnorr, 1994) is a famous variant of LLL lattice basis reduction algorithm, and PotLLL algorithm (Fontein et al., 2014) and $S^2$LLL algorithm (Yasuda and Yamaguchi, 2019) are recent polynomial-time variants of DeepLLL algorithm developed from cryptographic applications.
However, the known polynomial bounds for computational complexity are shown only for parameter $\delta < 1$; for \lq\lq optimal\rq\rq{} parameter $\delta = 1$ which ensures the best output quality, no polynomial bounds are known, and except for LLL algorithm, it is even not formally proved that the algorithm always halts within finitely many steps.
In this paper, we prove that these four algorithms always halt also with optimal parameter $\delta = 1$, and furthermore give explicit upper bounds for the numbers of loops executed during the algorithms.
Unlike the known bound (Akhavi, 2003) applicable to LLL algorithm only, our upper bounds are deduced in a unified way for all of the four algorithms.
\end{abstract}

\section{Introduction}
\label{sec:introduction}

Lattice basis reduction is one of the most important kinds of algorithms from both theoretical and practical viewpoints.
Roughly speaking, given a basis of a lattice in finite-dimensional Euclidean space, a lattice basis reduction (or lattice reduction) algorithm aims at outputting a basis of the same lattice consisting of short and nearly orthogonal basis vectors.
Major applications of lattice reduction include security analysis of cryptosystems.
It has been shown \cite{Shor} that most of the currently deployed public key cryptosystems, such as RSA cryptosystem \cite{rsa} and elliptic curve cryptosystems \cite{ell1,ell2}, can be broken by quantum algorithms in polynomial time.
Accordingly, \lq\lq post-quantum\rq\rq{} cryptosystems that are still secure even after the development of large-scale quantum computers have been an intensively studied topic in the area of cryptography.
Among them, lattice-based cryptography \cite{Ajtai-Dwork} is one of the main design principles, which is based on the computational hardness of the shortest vector problem (SVP) and/or the closest vector problem (CVP) on high-dimensional lattices.
Lattice reduction is a fundamental building block of the known algorithms for SVP and CVP, therefore it is important from the viewpoint of cryptography to evaluate the computational complexity of lattice reduction.

LLL algorithm \cite{lll} is the most famous lattice reduction algorithm, which is known to output, in polynomial time with respect to the dimension of the lattice, a lattice basis having a certain good property (see \cite{nguyenval: lll} for the details).
On the other hand, its variant named DeepLLL algorithm has been proposed \cite{deeplll} and is known to output a lattice basis with better property than the case of LLL algorithm.
However, it is not proved that DeepLLL algorithm in general halts within polynomial time.
To resolve this drawback, some variants of DeepLLL algorithm with (provably) polynomial computational complexity are actively studied mainly in the area of cryptography, such as PotLLL algorithm \cite{fonschnwag: potlll} and $S^2$LLL algorithm \cite{yasyam: s2lll}.

In LLL algorithm and its variants mentioned above, the guaranteed quality of the output basis is controlled by a parameter commonly denoted by $\delta$ (it is denoted by $\eta$ in \cite{yasyam: s2lll}, but here we unify the notation into the $\delta$) in a way that the output basis becomes better when $\delta$ becomes larger. 
The range of parameter is $1/4 < \delta \leq 1$ for LLL, DeepLLL, and PotLLL algorithms, while $0 < \delta \leq 1$ for $S^2$LLL algorithm.
Therefore, purely from the viewpoint of guaranteed output quality, the parameter choice $\delta = 1$ is optimal.
However, as opposed to the case $\delta < 1$ where polynomial (in the dimension of the lattice) upper bounds for the computational complexity have been given for LLL, PotLLL, and $S^2$LLL algorithms, no such polynomial bounds are known for the optimal case $\delta = 1$.
In fact, for DeepLLL, PotLLL, and $S^2$LLL algorithms with $\delta = 1$, even it has not been proved (to the authors' best knowledge) that these algorithms always halt within finitely many steps.
(See below for more details.)
In order to estimate the best possible performances of those algorithms (for the sake of e.g., closely analyzing the security of lattice-based cryptosystems), it is worthy to evaluate the complexity of those algorithms with optimal parameter $\delta = 1$ in detail.

\subsection{Our Contributions}
\label{sec:introduction__contributions}

In this paper, we prove that all of LLL algorithm and its variants mentioned above (i.e., DeepLLL, PotLLL, and $S^2$LLL algorithms) with optimal parameter $\delta = 1$ always halt within finitely many steps, and furthermore give explicit upper bounds for the numbers of loops executed during the algorithms.

Before explaining our main result, we introduce some notations (see Section \ref{sec:preliminaries_notations} for the definitions).
Let $n \geq 2$ be the dimension of an input lattice $L$.
We suppose that $L$ is full-rank for the sake of simplicity.
Let $M$ denote the maximum ($L^2$, a.k.a.\ Euclidean) norm of vectors in an input basis $\matfont{B}$ of the lattice $L$, and let $\vol(L)$ denote the volume of $L$.
We put $\alpha = M^n / \vol(L)$ (note that $\alpha \geq 1$).
Then our main result is stated as follows (which is the same as Theorem \ref{thm:number_of_while} below).

\begin{theorem}
\label{thm:main_intro}
In each of LLL algorithm and its variants mentioned above with parameter $\delta = 1$, the total number of main loops executed during the algorithm is at most
\[
(n - 1) \cdot \left( 3 \prod_{j=2}^{n-1} (2 + \sqrt{j+3}) \right)^n \alpha^{n-1} \enspace,
\]
which is also upper bounded by
\begin{align*}
&(n - 1) \cdot \left( \frac{ 2 }{ \sqrt{5} } + 1 \right)^{n(n-2)} \alpha^{n-1} \left( \frac{ 3 \cdot (n+2)! }{ 8 } \right)^{n/2} \\
&\leq \alpha^{n-1} \left( \left( \frac{ 2 }{ \sqrt{5} } + 1 \right) e^{-1/2} n^{1/2 + o(1)} \right)^{n^2} \mbox{ (when $n \to \infty$)} \enspace.
\end{align*}
In particular, these algorithms always halt within finitely many steps.
\end{theorem}

We give a proof of this theorem in Section \ref{sec:halt_property__results}.
As this bound is hyperexponential in $n$, it is intuitively estimated that the computational complexity (not just the number of loops) also has the same asymptotic bound, assuming that each loop can be executed in polynomial time with respect to $n$.

Intuitively speaking, our proof focuses on the index $k$ for the main loop of each of the four algorithms, which is either increased or decreased at each of the loop.
The algorithm starts from $k = 2$, and it halts when $k$ becomes larger than $n$.
The given lattice basis $\matfont{B}$ is not changed when $k$ is increased, while we can show that some earlier part of $\matfont{B}$ is shrinked when $k$ is decreased.
Now the theorem is proved by a recursive argument based on the fact that the number of lattice points within a given threshold of norms is explicitly upper bounded (which we show in Section \ref{sec:halt_property__lattice_points}).
One of the main advantages of our result is that \emph{essentially the same proof is applied to all of the four algorithms}, in contrast to the known upper bound in \cite{akh: 1lll} which is only applicable to LLL algorithm (see Section \ref{sec:introduction__related} for details).
Another advantage of our result is that we obtain an \emph{explicit} bound for the number of loops in the algorithms, not only an asymptotic bound in the limit case $n \to \infty$.
On the other hand, our theorem gives the first upper bounds for the case of DeepLLL, PotLLL, and $S^2$LLL algorithms with $\delta = 1$, but our bound does not improve the known bound in \cite{akh: 1lll} for LLL algorithm with $\delta = 1$.

Although all of the known theoretical upper bounds for the case $\delta = 1$ including our result are hyperexponential in $n$ which is essentially different from the polynomial bounds for the case $\delta < 1$, it is expected that the practical complexity for the case $\delta = 1$ would be not drastically larger than the case $\delta < 1$; we also give some experimental observation in Section \ref{sec:experiment}.
Hence, it is an important future research topic to improve the known theoretical upper bounds for the case $\delta = 1$.

\subsection{Related Work}
\label{sec:introduction__related}

Considering the optimal parameter $\delta = 1$, Akhavi \cite{akh: 1lll} gave an asymptotic upper bound $O(A^{n^3}\log M)$ for the computational complexity of LLL algorithm, where $A$ is an arbitrary constant with $A > (4/3)^{1/12}$.
This is the state-of-the-art result, to the authors' best knowledge.
Our bound in this paper is basically worse than Akhavi's bound unless $M$ (hence $\alpha$) is significantly small.
However, Akhavi's approach is specific to the case of LLL algorithm.
In more detail, roughly speaking, Akhavi's approach is focusing on the timing of exchanging the first vector $\vecfont{b}_1$ in the basis $\matfont{B}$ and analyzing the behavior of the volume of the sublattice spanned by the first few basis vectors, and is crucially based on the typical property of LLL algorithm that the exchange of basis vectors is performed only for two consecutive vectors.
As this property is not satisfied by DeepLLL algorithm and its variants, it seems difficult to extend Akhavi's approach to the case of the latter algorithms.
In contrast, our approach proposed in this paper is applicable in a unified way to all of the four algorithms considered in this paper.
It is expected that the high flexibility of our approach would enable us to deduce a similar upper bound when some further variant of (Deep)LLL algorithm will be developed.

On the other hand, it is mentioned in Footnote 6 of \cite{akh: 1lll} that LLL algorithm (with $\delta = 1$) \emph{for input basis consisting of integer vectors} has computational complexity $O(M^{n^2})$, which is again already better than our upper bound unless $\alpha$ is significantly small.
This bound is based on the fact that the quantity $\Pot(\matfont{B})$ for the basis $\matfont{B}$ (see Eq.\eqref{eq: pot} below for the definition) is monotonically decreasing during LLL algorithm and $\Pot(\matfont{B})$ is a positive integer for any basis $\matfont{B}$ consisting of integer vectors.
We note that the same upper bound for integer basis vectors is also available for PotLLL algorithm, as $\Pot(\matfont{B})$ is again monotonically decreasing during the algorithm.
However, $\Pot(\matfont{B})$ is not monotonically decreasing in DeepLLL and $S^2$LLL algorithms (we give examples of this fact in Section \ref{sec:halt_property__comparison}), therefore the same approach does not work for those algorithms.
In contrast, our approach in this paper is generally applicable without assuming that the input basis consists of integer vectors nor assuming that $\Pot(\matfont{B})$ is monotonically decreasing (or there is some integer quantity that is monotonically decreasing during the algorithm).

\section{Preliminaries}

\subsection{Notations and Basic Definitions}
\label{sec:preliminaries_notations}

In this paper, $\| \vecfont{b} \|$ denotes the (Euclidean) norm of a real vector $\vecfont{b}$.
The lattice generated by linearly independent vectors $\vecfont{b}_1,\cdots,\vecfont{b}_n\in\setR^m$ is defined by
\[
L(\vecfont{b}_1,\cdots, \vecfont{b}_n) = \set{\sum_{i=1}^n v_i \vecfont{b}_i}{v_i\in\setZ \mbox{\quad ($i=1,\cdots,n$)} } \enspace.
\]
The matrix $\matfont{B} = [\vecfont{b}_1,\cdots,\vecfont{b}_n]\in\setR^{m\times n}$ indicates a basis of the lattice $L = L(\vecfont{b}_1,\cdots, \vecfont{b}_n)$.
The tuple of vectors obtained by Gram--Schmidt orthogonalization (GSO) for $\matfont{B}$ is denoted by $\matfont{B}^\ast=[\vecfont{b}_1^\ast,\cdots,\vecfont{b}_n^\ast]$.
Namely, we have
\[
\vecfont{b}_i^\ast = \vecfont{b}_i - \sum_{j=1}^{i-1}\mu_{ij} \vecfont{b}_j^\ast \mbox{ for } i = 1,\cdots, n \mbox{, where } \mu_{ij} = \frac{\langle \vecfont{b}_i, \vecfont{b}_j^\ast\rangle}{\|\vecfont{b}_j^\ast\|^2} \mbox{ for } 1\leq j < i \leq n
\]
(here, as usual, the summation $\sum_{i=p}^q$ becomes zero when $p > q$; and $\langle\cdot,\cdot\rangle$ denotes the standard inner product for vectors).
Set $B_i = \|\vecfont{b}_i^\ast\|^2$ for $i=1,\cdots,n$.

The volume of the lattice $L$ with basis $\matfont{B}$ as above is given by
\begin{equation}
\label{eq: lattice_volume}
\vol(L) = \sqrt{{}^t \matfont{B} \cdot \matfont{B}} = \prod_{i=1}^n \|\vecfont{b}_i^\ast\| \enspace,
\end{equation}
where ${}^t \matfont{B}$ denotes the transpose of a matrix $\matfont{B}$.
Let $\pi_j$ be the orthogonal projection from $\setR^m$ to the orthogonal complement of the subspace $\Span\{\vecfont{b}_1,\cdots, \vecfont{b}_{j-1}\}$; in particular, we have $\pi_1=\id$.
Note that the vector $\vecfont{b}_i^\ast$ obtained by GSO as above is equal to $\pi_i(\vecfont{b}_i)$.
We define a permutation $\sigma_{i,k}$ for the basis vectors by
\[
\sigma_{i,k}(\matfont{B})=[\vecfont{b}_1,\cdots, \vecfont{b}_{i-1},\vecfont{b}_k,\vecfont{b}_{i},\cdots,\vecfont{b}_{k-1}, \vecfont{b}_{k+1},\cdots,\vecfont{b}_n] \enspace.
\]

In the rest of this paper, to simplify the arguments, we focus only on full-rank lattices, i.e., the case of $n = m$.
Our results might be easily extendible to the general case where $n \leq m$.
Our argument below is based on the well-known discreteness of lattices; for $r \geq 0$, the number of points $\vecfont{b} \in L$ with $\| \vecfont{b} \| \leq r$ is finite (its refinement will be given in Lemma \ref{lem:points} later).

\subsection{The LLL Algorithm and Its Variants}

We recall some properties of LLL algorithm \cite{lll} and its variants such as DeepLLL algorithm \cite{deeplll}.
LLL algorithm takes a basis $\matfont{B}$ of a lattice (which is supposed to be full-rank in this paper) as input and has an auxiliary parameter $\delta \in (1/4,1]$ controlling a trade-off between the computational complexity and the quality of the output.
The algorithm outputs a new basis of the same lattice with a certain property called \emph{$\delta$-LLL-reduced}.
Here we omit the definition of $\delta$-LLL-reduced property as it is not relevant to our argument in this paper (see e.g., \cite{yasyam: s2lll} for the details), and only mention that if $\delta$ increases, then the $\delta$-LLL-reduced property becomes stronger while the upper bound for computational complexity of the algorithm is also getting larger.
In this paper, we write \emph{$\delta$-LLL} to mean LLL algorithm with parameter $\delta$.
The construction of LLL algorithm will be described in Section \ref{sec:halt_property__results} below.

Among the known variants of LLL algorithm, here we deal with DeepLLL algorithm \cite{deeplll}, PotLLL algorithm \cite{fonschnwag: potlll}, and $S^2$LLL algorithm \cite{yasyam: s2lll} (see Section \ref{sec:halt_property__results} for their constructions).
Similarly to the case of LLL, those algorithms also have their own trade-off parameter $\delta$; the range is $1/4 < \delta \leq 1$ for DeepLLL and PotLLL, and $0 < \delta \leq 1$ for $S^2$LLL (note that the parameter is denoted by $\eta$ in the original paper \cite{yasyam: s2lll}).
We use the words \emph{$\delta$-DeepLLL}, \emph{$\delta$-PotLLL}, and \emph{$\delta$-$S^2$LLL} in a way similar to $\delta$-LLL.
Definitions of the properties satisfied by the outputs of those algorithms are also omitted herein.

Roughly speaking, DeepLLL is modified from LLL in a way that DeepLLL may move a basis vector to a position not adjacent to the original while LLL only exchanges two adjacent basis vectors.
In comparison to DeepLLL, the strategy of PotLLL is that permutations for the basis $\matfont{B}$ are iterated in the direction of decreasing the following quantity:
\begin{equation}
\label{eq: pot}
\Pot(\matfont{B})
=\prod_{i=1}^n \vol(L(\vecfont{b}_1,\cdots, \vecfont{b}_i))^2
=\prod_{i=1}^n B_i^{n-i+1} \enspace.
\end{equation}
We note that when $\matfont{B}$ consists of integer vectors, this value is always a non-negative integer, as each $\vol(L(\vecfont{b}_1,\cdots, \vecfont{b}_i))^2$ is an integer due to Eq.\eqref{eq: lattice_volume}.
We also note the following relation for $i < k$ (see Lemma 1 of \cite{fonschnwag: potlll}):
\begin{equation}
\label{eq:pot_for_SWAP}
\Pot(\sigma_{i,k}(\matfont{B})) = \Pot(\matfont{B})\prod_{j=i}^{k-1}\frac{\|\pi_j(\vecfont{b}_k)\|^2}{B_j} \enspace.
\end{equation}
On the other hand, the iteration in $S^2$LLL is performed in a way that the following quantity, instead of $\Pot(\matfont{B})$, is monotonically decreased:
\[
\mss(\matfont{B}) = \sum_{i=1}^n B_i \enspace.
\]
We note the following relation for $i < k$ (see Eq.(5) of \cite{yasyam: s2lll}):
\begin{equation}
\label{eq:s2_for_SWAP}
\begin{split}
S_{ik}&:=\mss(\matfont{B})-\mss(\sigma_{i,k}(\matfont{B})) \\
&= \sum_{j=i}^{k-1}\mu_{kj}^2B_j\left(\frac{B_j}{\|\pi_j(\vecfont{b}_k)\|^2}-1\right) \enspace.
\end{split}
\end{equation}

\section{Our Complexity Evaluation}
\label{sec:halt_property}

In this section, we prove that LLL algorithm and its aforementioned variants with the largest possible parameter $\delta = 1$ always halt within finitely many steps, and also provide upper bounds for the number of loops executed before the algorithm halts.
Although the halt property for the case of LLL algorithm has been known, our proof strategy is significantly different from the known proof in \cite{akh: 1lll} and is applicable also to the variants of LLL algorithm (in contrast to the strategy of \cite{akh: 1lll} specific to the case of LLL algorithm).

In the constructions of those algorithms described below, \lq\lq Size-reduce\rq\rq{} for basis $\matfont{B} = [\vecfont{b}_1,\cdots,\vecfont{b}_n]$ means the procedure to update $\vecfont{b}_i \leftarrow \vecfont{b}_i - \lfloor \mu_{ij} \rceil \vecfont{b}_j$ for $i = 1,\cdots,n$ and $j = 1,\cdots,i-1$ (in this order), where $\lfloor \,\cdot\, \rceil$ denotes rounding to the nearest integer; see Algorithm 24 of \cite{galbraith: size-reduce} for the details.
We note that after the update, the value of $\mu_{ij}$ becomes $\mu_{ij} - \lfloor \mu_{ij} \rceil$, which thus has absolute value at most $1/2$.
(Such a basis is said to be \emph{size-reduced}.)
\lq\lq Update GSO\rq\rq{} means the procedure to also update the tuple $\matfont{B}^\ast$ associated to the updated basis $\matfont{B}$.
Moreover, in the proofs below, $\mathrm{SWAP}_{\rho,\kappa}$ means the step of $\matfont{B}\leftarrow \sigma_{\rho,\kappa}(\matfont{B})$ in the algorithm.

\subsection{Number of Lattice Points with Bounded Norm}
\label{sec:halt_property__lattice_points}

Our proof below is based on the fact that the number of lattice points with bounded norm is finite.
Here we give the following refinement of this fact.

\begin{lemma}
\label{lem:points}
Let $L$, $\matfont{B} = [\vecfont{b}_1,\cdots, \vecfont{b}_n]$, and $\matfont{B}^\ast = [\vecfont{b}_1^\ast,\cdots, \vecfont{b}_n^\ast]$ be as in Section \ref{sec:preliminaries_notations}.
For $r \geq 0$, the number of points $\vecfont{x} \in L$ with $\| \vecfont{x} \| \leq r$ is at most
\[
\prod_{i=1}^n \bracket{\frac{2r}{\|\vecfont{b}_i^\ast\|}+1} \leq \frac{(M + 2r)^n}{\vol(L)}
\]
where $M = \max\{ \| \vecfont{b}_1 \|, \cdots, \| \vecfont{b}_n \| \}$.
\end{lemma}

\begin{proof}
Let $\vecfont{x} = \sum_{i=1}^n v_i \vecfont{b}_i$ ($v\in\setZ^n$) be a lattice point satisfying $\|\vecfont{x}\|\leq r$. Since
\begin{align*}
\|\vecfont{x}\|^2 &= \norm{\sum_{i=1}^n v_i\bracket{\vecfont{b}_i^\ast + \sum_{j=1}^{i-1}\mu_{ij}\vecfont{b}_j^\ast}}^2 \\
&= \norm{\sum_{j=1}^n \bracket{\sum_{i=j+1}^{n} v_i\mu_{ij} + v_j}\vecfont{b}_j^\ast}^2 \\
&= \sum_{j=1}^n \bracket{\sum_{i=j+1}^{n} v_i\mu_{ij} + v_j}^2 \|\vecfont{b}_j^\ast\|^2 \enspace,
\end{align*}
we have
\[
\abs{\sum_{i=j+1}^{n} v_i\mu_{ij} + v_j}  \leq \frac{r}{\|\vecfont{b}_j^\ast\|} \mbox{ for } j = 1,\cdots, n \enspace.
\]

We give an upper bound for the number of $v \in \setZ^n$ satisfying the inequality above.
First, as $\abs{v_n} \leq r/\|\vecfont{b}_n^\ast\|$, $v_n$ has at most ${2\,r/\|\vecfont{b}_n^\ast\|} +1$ possibilities.
On the other hand, when $v_n,v_{n-1},\cdots, v_{k+1}$ are fixed, we have a condition $\abs{\gamma + v_k} \leq r / \| \vecfont{b}_k^\ast \|$ where $\gamma = \sum_{i = k+1}^{n} v_i \mu_{ik}$ which is also a fixed value, and hence $v_k$ has at most $2r / \|\vecfont{b}_k^\ast\| +1$ possibilities.
By repeating this, the number of those $v\in\setZ^n$ is at most $\prod_{i=1}^n (2r / \|\vecfont{b}_i^\ast\| + 1)$.

Now the inequality in the statement holds since
\[
\prod_{i=1}^n \bracket{\frac{2r}{\|\vecfont{b}_i^\ast\|}+1}
= \frac{\prod_{i=1}^n(2r+\|\vecfont{b}_i^\ast\|)}{\prod_{i=1}^n\|\vecfont{b}_i^\ast\|}
\leq \frac{(2r + M)^n}{\vol(L)} \enspace.
\]
Hence the lemma holds.
\end{proof}

\subsection{The Results}
\label{sec:halt_property__results}

We recall the construction of LLL algorithm, DeepLLL algorithm, PotLLL algorithm, and $S^2$LLL algorithm as in Algorithms \ref{algo: lll}, \ref{algo: deep}, \ref{algo: pot}, and \ref{algo: s2}, respectively.
Here, for the $\argmin$ in PotLLL algorithm and the $\argmax$ in $S^2$LLL algorithm, we suppose that the largest index achieving the minimum and the maximum values, respectively, is chosen when more than one candidates exist.

 \begin{algorithm}
 \caption{LLL Algorithm}
 \label{algo: lll}
 \begin{algorithmic}[1]
 \REQUIRE $\matfont{B}$ and parameter $1/4<\delta\leq1$
 \ENSURE  $\delta$-LLL-reduced basis $\matfont{B}$
 %\\ \textit{Initialisation} :
  \STATE $k\leftarrow 2$
% \\ \textit{LOOP Process}
 \WHILE {$k\leq n$}
  \STATE Size-reduce $\matfont{B}=[\vecfont{b}_1,\cdots,\vecfont{b}_n]$ \label{algo:lll_reduce}
  \IF {$B_k \geq (\delta - \mu_{k,k-1}^2)B_{k-1}$} \label{algo:lll_if}
  	\STATE $k \leftarrow k+1$ \label{algo:lll_increment}
  \ELSE
    \STATE $\matfont{B}\leftarrow \sigma_{k-1,k}(\matfont{B})$, Update GSO \label{algo:lll_swap}
    \STATE $k \leftarrow \max\{k-1,2\}$
    \STATE goto Step 3
  \ENDIF
 \ENDWHILE
 \RETURN $\matfont{B}$ 
 \end{algorithmic} 
 \end{algorithm}

 \begin{algorithm}
 \caption{DeepLLL Algorithm}
 \label{algo: deep}
 \begin{algorithmic}[1]
 \REQUIRE $\matfont{B}$ and parameter $1/4<\delta\leq1$
 \ENSURE  $\delta$-DeepLLL-reduced basis $\matfont{B}$
 %\\ \textit{Initialisation} :
  \STATE $k\leftarrow 2$
% \\ \textit{LOOP Process}
 \WHILE {$k\leq n$}
  \STATE Size-reduce $\matfont{B}$ \label{algo:deep_reduce}
  \STATE $C\leftarrow \|\vecfont{b}_k\|^2$
  \FOR {$i = 1, \cdots, k-1$}
    \IF {$C \geq \delta \|\vecfont{b}_i^\ast\|^2 $} \label{algo:deep_if}
      \STATE $C \leftarrow C - \mu_{ki}^2 \|\vecfont{b}_i^\ast\|^2$
    \ELSE
      \STATE $\matfont{B}\leftarrow \sigma_{i,k}(\matfont{B})$, Update GSO
      \STATE $k \leftarrow \max\{i, 2\}$
      \STATE goto Step 3
    \ENDIF
  \ENDFOR
  \STATE $k \leftarrow k+1$ \label{algo:deep_increment}
 \ENDWHILE
 \RETURN $\matfont{B}$ 
 \end{algorithmic} 
 \end{algorithm}

 \begin{algorithm}
 \caption{PotLLL Algorithm}
\label{algo: pot}
 \begin{algorithmic}[1]
 \REQUIRE $\matfont{B}$ and parameter $1/4<\delta\leq1$
 \ENSURE  $\delta$-PotLLL-reduced basis $\matfont{B}$
 %\\ \textit{Initialisation} :
  \STATE $k\leftarrow 2$
% \\ \textit{LOOP Process}
 \WHILE {$k\leq n$}
  \STATE Size-reduce $\matfont{B}$ \label{algo:pot_reduce}
  \STATE $i\leftarrow \argmin_{1\leq j\leq k-1}\Pot(\sigma_{j,k}(\matfont{B}))$
  \IF {$\delta\Pot(\matfont{B})\leq\Pot(\sigma_{i,k}(\matfont{B}))$} \label{algo:pot_if}
  	\STATE $k \leftarrow k+1$ 
  \ELSE
    \STATE $\matfont{B}\leftarrow \sigma_{i,k}(\matfont{B})$, Update GSO
    \STATE $k \leftarrow \max\{i,2\}$
    \STATE goto Step 3
  \ENDIF
 \ENDWHILE
 \RETURN $\matfont{B}$ 
 \end{algorithmic} 
 \end{algorithm}
 
  \begin{algorithm}
 \caption{$S^2$LLL Algorithm}
 \label{algo: s2}
 \begin{algorithmic}[1]
 \REQUIRE $\matfont{B}$ and parameter $0<\delta\leq1$
 \ENSURE  $\delta$-$S^2$LLL-reduced basis $\matfont{B}$
 %\\ \textit{Initialisation} :
  \STATE $k\leftarrow 2$
% \\ \textit{LOOP Process}
 \WHILE {$k\leq n$}
  \STATE Size-reduce $\matfont{B}$ \label{algo:s2_reduce}
  \STATE $i\leftarrow \argmax_{1\leq j\leq k-1}S_{jk}$
  \IF {$S_{ik}\leq(1-\delta)\mss(\matfont{B})$} \label{algo:s2_if}
  	\STATE $k \leftarrow k+1$ 
  \ELSE
    \STATE $\matfont{B}\leftarrow \sigma_{i,k}(\matfont{B})$, Update GSO
    \STATE $k \leftarrow \max\{i,2\}$
    \STATE goto Step 3
  \ENDIF
 \ENDWHILE
 \RETURN $\matfont{B}$ 
 \end{algorithmic} 
 \end{algorithm}

Let $M$ denote the value $\max\{ \| \vecfont{b}_1 \|, \cdots, \| \vecfont{b}_n \| \}$ at the beginning of the algorithm.
We first note the following properties.

\begin{lemma}
\label{lem:norm_bound}
At any step of these algorithms, we have $\max\{ \| \vecfont{b}_1 \|, \cdots, \| \vecfont{b}_n \| \} \leq M$.
\end{lemma}
\begin{proof}
By the construction of the algorithms, the content of $\matfont{B}$ is changed only by either size-reduction or $\mathrm{SWAP}_{i,k}$ for some $i,k$.
Now the norm of each of $\vecfont{b}_1, \cdots, \vecfont{b}_n$ is not increased by size-reduction by the definition, while $\max\{ \| \vecfont{b}_1 \|, \cdots, \| \vecfont{b}_n \| \}$ is not changed by $\mathrm{SWAP}_{i,k}$ as it just permutes the elements in $\matfont{B}$.
Hence the claim holds.
\end{proof}

\begin{lemma}
\label{lem:B_is_decreased}
For operations $\mathrm{SWAP}_{\rho,k}$ and $\mathrm{GSO}$ in the algorithms, let $B_{\rho}^{\mathrm{new}}$ denote the value of $B_{\rho}$ after the update.
Then $B_{\rho}^{\mathrm{new}} < B_{\rho}$.
\end{lemma}
\begin{proof}
We also use the superscript \lq\lq $\mathrm{new}$\rq\rq{} for other objects in a similar way.
For LLL algorithm, $\mathrm{SWAP}_{\rho,k}$ is executed only for $\rho = k - 1$ and when $B_k < (\delta - \mu_{k,k-1}^2) B_{k-1}$.
The latter condition is equivalent to
\[
\mbox{($B_{k-1} \geq$) } \delta B_{k-1} > B_k + \mu_{k,k-1}^2 B_{k-1} \mbox{ ($= B_{k-1}^{\mathrm{new}}$)}
\]
(note that for $\vecfont{b}_k = \vecfont{b}_k^\ast + \sum_{j=1}^{k-1} \mu_{kj} \vecfont{b}_j^\ast$, we have $\mu_{kj}^{\mathrm{new}} = \langle \vecfont{b}_k,\vecfont{b}_j^{\ast} \rangle B_j^{-2} = \mu_{kj}$ for $j \leq k-2$, therefore $(\vecfont{b}_{k-1}^\ast)^{\mathrm{new}} = \vecfont{b}_k^{\ast} + \mu_{k,k-1} \vecfont{b}_{k-1}^{\ast}$ and $B_{k-1}^{\mathrm{new}} = B_k + \mu_{k,k-1}^2 B_{k-1}$).
Hence the claim holds for the case of LLL algorithm.

For DeepLLL algorithm, when $\mathrm{SWAP}_{\rho,k}$ is executed, the value of $C$ at Step \ref{algo:deep_if} with $i = \rho$ is $\| \vecfont{b}_k \|^2 - \sum_{j=1}^{\rho-1} \mu_{kj}^2 B_j$.
Therefore, the negation of the condition in Step \ref{algo:deep_if} with $i = \rho$ is
\[
\mbox{($B_{\rho} \geq$) } \delta B_{\rho} > \| \vecfont{b}_k \|^2 - \sum_{j=1}^{\rho-1} \mu_{kj}^2 B_j \mbox{ ($= B_{\rho}^{\mathrm{new}}$)} \enspace.
\]
Hence the claim holds for the case of DeepLLL algorithm.

For PotLLL algorithm, when $\mathrm{SWAP}_{\rho,k}$ is executed, we have $\Pot(\sigma_{\rho+1,k}(\matfont{B})) > \Pot(\sigma_{\rho,k}(\matfont{B}))$ where we put $\sigma_{k,k}(\matfont{B}) = \matfont{B}$ for the case $\rho = k - 1$; this follows from the negation of the condition in Step \ref{algo:pot_if} when $\rho = k - 1$, and follows from the choice of the maximal index in the $\argmin$ when $\rho \leq k - 2$.
By Eq.\eqref{eq:pot_for_SWAP}, this implies that
\[
1 > \frac{ \| \pi_{\rho}(\vecfont{b}_k) \|^2 }{ B_{\rho} } \enspace,
\]
i.e., $B_{\rho} > \| \pi_{\rho}(\vecfont{b}_k) \|^2 = B_{\rho}^{\mathrm{new}}$.
Hence the claim holds for the case of PotLLL algorithm.

For $S^2$LLL algorithm, when $\mathrm{SWAP}_{\rho,k}$ is executed, for the case $\rho = k - 1$, it follows from the negation of the condition in Step \ref{algo:s2_if} that $S_{k-1,k} > (1 - \delta) \mss(\matfont{B})$ ($\geq 0$).
On the other hand, for the case $\rho \leq k - 2$, by the choice of the maximal index in the $\argmax$, we have $S_{\rho k} > S_{\rho+1, k}$.
In both cases, by Eq.\eqref{eq:s2_for_SWAP}, we have
\[
\mu_{k \rho}^2 B_{\rho} \left( \frac{ B_{\rho} }{ \| \pi_{\rho}(\vecfont{b}_k) \|^2 } - 1 \right) > 0 \enspace,
\]
therefore $B_{\rho} > \| \pi_{\rho}(\vecfont{b}_k) \|^2 = B_{\rho}^{\mathrm{new}}$.
Hence the claim holds for the case of $S^2$LLL algorithm.
\end{proof}

The main observation in this section is the following.
Here we put $\alpha = M^n / \vol(L) \geq 1$.

\begin{lemma}
\label{lem:number_of_while}
Let $w_1,\cdots, w_{n-1}$ be any sequence of real numbers satisfying
\[
w_1 \geq 3^n \alpha - 2 \,,\,
w_k \geq \left( (1 + \sqrt{ k + 3 })^n \alpha - 2 \right) \cdot \left( 1 + \sum_{j = 1}^{k - 1} w_j \right) \mbox{ for } k = 2,\cdots, n - 1 \enspace.
\]
Then the total number of \emph{\texttt{while}} blocks in each of the algorithms is at most
\[
(n - 1) \cdot \left( 1 + \sum_{j=1}^{n-1} w_j \right) + 1 \enspace.
\]
\end{lemma}
\begin{proof}
For $1 \leq k_0 \leq n - 1$, let $\mathcal{W}_{k_0}$ denote the set of \texttt{while} blocks in which $\mathrm{SWAP}_{k_0,k}$ for some $k$ is executed.
For $0 \leq k_0 \leq n - 1$, let $\mathcal{W}_{\leq k_0} = \bigcup_{\kappa = 1}^{k_0} \mathcal{W}_{\kappa}$ (note that $\mathcal{W}_0 = \emptyset$).

Let $1 \leq k_0 \leq n - 1$.
We consider any time interval in the algorithm in which no \texttt{while} blocks in $\mathcal{W}_{\leq k_0 - 1}$ appear, and evaluate the number, say $N_{k_0}$, of \texttt{while} blocks in $\mathcal{W}_{k_0}$ in the interval.
Now $\mathrm{SWAP}_{i,k}$ with $i \leq k_0 - 1$ is not executed during the interval, therefore the values of $\vecfont{b}_1,\cdots, \vecfont{b}_{k_0 - 1}$ and $B_1,\cdots, B_{k_0 - 1}$ are not changed during the interval.
Moreover, by Lemma \ref{lem:B_is_decreased}, the value of $B_{k_0}$ is monotonically decreased at each \texttt{while} block in $\mathcal{W}_{k_0}$.
Therefore, at each \texttt{while} block in $\mathcal{W}_{k_0}$ during the interval, $\vecfont{b}_{k_0}$ is changed to some vector that did not previously appear as the value of $\vecfont{b}_{k_0}$.
This implies that there are at least $N_{k_0} + 1$ different possibilities of the non-zero lattice point $\vecfont{b}_{k_0}$.
On the other hand, as
\[
\| \vecfont{b}_{k_0} \|^2
= B_{k_0} + \sum_{j=1}^{k_0-1} \mu_{k_0 j}{}^2 B_j
\leq B_{k_0} + \frac{1}{4}\sum_{j=1}^{k_0-1} B_j
\leq \frac{ k_0 + 3 }{ 4 } \cdot M^2
\]
(see Lemma \ref{lem:norm_bound}), Lemma \ref{lem:points} implies that there are at most
\[
\frac{ (M + 2 \cdot M \sqrt{ (k_0 + 3)/4 })^n }{ \vol(L) } - 1
= (1 + \sqrt{ k_0 + 3 })^n \alpha - 1
\]
possibilities of the non-zero lattice point $\vecfont{b}_{k_0}$.
Hence we have
\[
N_{k_0} + 1 \leq (1 + \sqrt{ k_0 + 3 })^n \alpha - 1 \enspace,
\]
therefore
\[
N_{k_0} \leq (1 + \sqrt{ k_0 + 3 })^n \alpha - 2 \enspace.
\]

By applying the result above with $k_0 = 1$ to the whole algorithm, it follows that
\[
|\mathcal{W}_1| \leq N_1 \leq 3^n \alpha - 2 \enspace.
\]
On the other hand, let $2 \leq k_0 \leq n - 1$ and consider the interval between two consecutive \texttt{while} blocks in $\mathcal{W}_{\leq k_0 - 1}$, or the interval from the beginning of the algorithm until the first \texttt{while} block in $\mathcal{W}_{\leq k_0 - 1}$, or the interval from the end of the last \texttt{while} block in $\mathcal{W}_{\leq k_0 - 1}$ until the end of the algorithm.
By applying the result above to the interval, it follows that there are at most $N_{k_0}$ \texttt{while} blocks in $\mathcal{W}_{k_0}$ in the interval.
As there are $|\mathcal{W}_{\leq k_0 - 1}| + 1$ such intervals, it follows that
\[
|\mathcal{W}_{k_0}|
\leq N_{k_0} \cdot (|\mathcal{W}_{\leq k_0 - 1}| + 1)
\leq \left( (1 + \sqrt{ k_0 + 3 })^n \alpha - 2 \right) \cdot \left( 1 + \sum_{j = 1}^{k_0 - 1} |\mathcal{W}_j| \right) \enspace.
\]
Hence it holds recursively that $|\mathcal{W}_{k_0}| \leq w_{k_0}$ for any $1 \leq k_0 \leq n - 1$.

Among the \texttt{while} blocks in the algorithm, those not in $\mathcal{W}_{\leq n-1}$ satisfy that $k$ is incremented in the \texttt{while} block.
As there are at most $n - 1$ consecutive such \texttt{while} blocks (except the last \texttt{while} block in the algorithm where the $k$ is incremented from $n$ to $n + 1$), the total number of \texttt{while} blocks in the algorithm is at most
\[
(n - 1) \cdot (|\mathcal{W}_{\leq n - 1}| + 1) + 1
\leq (n - 1) \cdot \left( 1 + \sum_{j=1}^{n-1} w_j \right) + 1 \enspace.
\]
Hence the claim holds.
\end{proof}

Now our main result is given as follows (this is the same as the theorem described in the introduction).

\begin{theorem}
\label{thm:number_of_while}
Assume that $n \geq 2$.
In each of LLL algorithm and its variants above, including the case of the maximal parameter $\delta = 1$, the total number of \emph{\texttt{while}} blocks is at most
\begin{align*}
(n - 1) \cdot \left( 3 \prod_{j=2}^{n-1} (2 + \sqrt{j+3}) \right)^n \alpha^{n-1}
&\leq (n - 1) \cdot \left( \frac{ 2 }{ \sqrt{5} } + 1 \right)^{n(n-2)} \alpha^{n-1} \left( \frac{ 3 \cdot (n+2)! }{ 8 } \right)^{n/2} \\
&\leq \alpha^{n-1} \left( \left( \frac{ 2 }{ \sqrt{5} } + 1 \right) e^{-1/2} n^{1/2 + o(1)} \right)^{n^2} \mbox{ (when $n \to \infty$)} \enspace,
\end{align*}
where $\alpha = M^n / \vol(L)$.
In particular, these algorithms always halt within finitely many steps.
\end{theorem}
\begin{proof}
First we show recursively that the sequence $w_1,\cdots, w_{n-1}$ given by
\[
w_1 = 3^n \alpha - 2 \,,\,
w_k = \left( 3 (1 + \sqrt{k+3}) \prod_{j=2}^{k-1} (2 + \sqrt{j+3}) \right)^n \alpha^k - 2 \mbox{ for } k \geq 2
\]
satisfies the condition in Lemma \ref{lem:number_of_while}.
The case of $w_1$ is obvious.
When $k \geq 2$ and the claim holds for $w_1,\cdots,w_{k-1}$, it follows recursively that
\[
\sum_{i=1}^m w_i \leq \left( 3 \prod_{j=2}^{m} (2 + \sqrt{j+3}) \right)^n \alpha^m - 2 \mbox{ for } m = 1,\cdots, k-1 \enspace.
\]
Indeed, this is obvious when $m = 1$.
When $m \geq 2$, we have
\begin{align*}
\sum_{i=1}^{m} w_i
&= \sum_{i=1}^{m-1} w_i + w_m \\
&\leq \left( 3 \prod_{j=2}^{m-1} (2 + \sqrt{j+3}) \right)^n \alpha^{m-1} + \left( 3 (1 + \sqrt{m+3}) \prod_{j=2}^{m-1} (2 + \sqrt{j+3}) \right)^n \alpha^m - 2 \\
&\leq \left( 3 \prod_{j=2}^{m-1} (2 + \sqrt{j+3}) \right)^n \alpha^m \left( 1 + (1 + \sqrt{m + 3})^n \right) - 2
\end{align*}
where we used the relation $\alpha \geq 1$.
By using the fact
\[
1 + (1 + \sqrt{m + 3})^n
\leq (1 + (1 + \sqrt{m+3}))^n
= (2 + \sqrt{m + 3})^n \enspace,
\]
we have
\[
\sum_{i=1}^{m} w_i
\leq \left( 3 \prod_{j=2}^{m-1} (2 + \sqrt{j+3}) \right)^n \alpha^m (2 + \sqrt{m + 3})^n - 2
= \left( 3 \prod_{j=2}^{m} (2 + \sqrt{j+3}) \right)^n \alpha^m - 2 \enspace,
\]
as desired.

By using the result above, the right-hand side of the inequality in Lemma \ref{lem:number_of_while} to be verified becomes
\begin{align*}
&\left( (1 + \sqrt{ k + 3 })^n \alpha - 2 \right) \cdot \left( 1 + \sum_{j = 1}^{k - 1} w_j \right) \\
&\leq \left( (1 + \sqrt{ k + 3 })^n \alpha - 2 \right) \cdot \left( \left( 3 \prod_{j=2}^{k-1} (2 + \sqrt{j+3}) \right)^n \alpha^{k-1} - 1 \right) \\
&\leq (1 + \sqrt{ k + 3 })^n \alpha \cdot \left( 3 \prod_{j=2}^{k-1} (2 + \sqrt{j+3}) \right)^n \alpha^{k-1} - 2
= w_k
\end{align*}
where we used the relations $(1 + \sqrt{ k + 3 })^n \alpha \geq 0$ and $\left( 3 \prod_{j=2}^{k-1} (2 + \sqrt{j+3}) \right)^n \alpha^{k-1} \geq 2$.
Hence the condition in Lemma \ref{lem:number_of_while} is satisfied by these $w_k$.
Therefore, the total number of \texttt{while} blocks is at most the value in the statement of Lemma \ref{lem:number_of_while}.
Now by the argument above, this value becomes
\begin{align*}
(n - 1) \cdot \left( 1 + \sum_{j=1}^{n-1} w_j \right) + 1
&\leq (n - 1) \cdot \left( \left( 3 \prod_{j=2}^{n-1} (2 + \sqrt{j+3}) \right)^n \alpha^{n-1} - 1 \right) + 1 \\
&\leq (n - 1) \cdot \left( 3 \prod_{j=2}^{n-1} (2 + \sqrt{j+3}) \right)^n \alpha^{n-1} \enspace,
\end{align*}
which is the left-hand side of the inequality in the statement of this theorem, as desired.

Moreover, by using the fact that $2 + \sqrt{j+3} \leq (2/\sqrt{5} + 1) \sqrt{j+3}$ for any $j \geq 2$, the value above is smaller than or equal to
\begin{align*}
&(n - 1) \cdot \left( 3 \left( \frac{ 2 }{ \sqrt{5} } + 1 \right)^{n-2} \cdot  \prod_{j=2}^{n-1} \sqrt{j+3} \right)^n \alpha^{n-1} \\
&= (n - 1) \cdot 3^n \left( \frac{ 2 }{ \sqrt{5} } + 1 \right)^{n(n-2)} \alpha^{n-1} \left( \frac{ (n+2)! }{ 4! } \right)^{n/2} \\
&= (n - 1) \cdot \left( \frac{ 2 }{ \sqrt{5} } + 1 \right)^{n(n-2)} \alpha^{n-1} \left( \frac{ 3 \cdot (n+2)! }{ 8 } \right)^{n/2} \enspace,
\end{align*}
as desired.
By using Stirling's Formula $m! \leq \sqrt{2\pi} m^{m + 1/2} e^{-m + 1/(12m)}$, this value is further bounded by
\[
(n - 1) \cdot \left( \frac{ 2 }{ \sqrt{5} } + 1 \right)^{n(n-2)} \alpha^{n-1} \left( \frac{ 3 \cdot \sqrt{2\pi} (n+2)^{n + 5/2} e^{- n - 2 + 1/(12(n+2))} }{ 8 } \right)^{n/2} \enspace,
\]
which at $n \to \infty$ is upper bounded by
\begin{align*}
&\left( \frac{ 2 }{ \sqrt{5} } + 1 \right)^{n(n-2)} \alpha^{n-1} \left( (n+2)^{n + 5/2 + o(1)} e^{-n} \right)^{n/2} \\
&= \alpha^{n-1} \left( \left( \frac{ 2 }{ \sqrt{5} } + 1 \right)^{1 - 2/n} (n+2)^{1/2 + 5/(4n) + o(1)} e^{-1/2} \right)^{n^2} \\
&\leq \alpha^{n-1} \left( \left( \frac{ 2 }{ \sqrt{5} } + 1 \right) e^{-1/2} n^{1/2 + o(1)} \right)^{n^2} \enspace.
\end{align*}
Hence the theorem holds.
\end{proof}

\begin{remark}
By assuming that each \texttt{while} block can be executed within polynomial time in $n$, it is intuitively estimated that the computational complexity of each of the algorithms has the same asymptotic bound as Theorem \ref{thm:number_of_while}, as the bound in Theorem \ref{thm:number_of_while} is already hyperexponential in $n$.
\end{remark}

\subsection{Comparisons to the Previous Results}
\label{sec:halt_property__comparison}

When $\vol(L)$ is almost constant, i.e., $\alpha \approx M^n$, the bound for the computational complexity of $1$-LLL, $1$-DeepLLL, $1$-PotLLL, and $1$-$S^2$LLL obtained by Theorem \ref{thm:number_of_while} can be roughly written as $(Mn)^{O(n^2)}$.
Comparing to the bound $O(A^{n^3}\log M)$ for $1$-LLL given in \cite{akh: 1lll} (where $A$ is a constant with $A > (4/3)^{1/12}$), our bound becomes better for the case of smaller $M$ (e.g., when $M = o(2^n/n)$), while our bound does not improve the previous one for the case of larger $M$ (e.g., when $M = \Omega(2^n)$).
The latter case may often occur in practical situations.
For example, for any basis $\matfont{B} = [\vecfont{b}_1,\cdots,\vecfont{b}_n]$ provided as a problem instance of Darmstadt SVP Challenge \cite{SVPChallenge}, $\matfont{B}$ is an upper triangular matrix where the first component of each column $\vecfont{b}_j$ is $2^{\Theta(n)}$, the diagonal components (except for the first column) are all $1$, and the other components are all $0$.
Now we have $M \sim \vol(L) = 2^{\Theta(n)}$ and our bound becomes $(2^{\Theta(n)} n)^{O(n^2)}$, which does not improve the bound $O(A^{n^3}\log M)$.
Moreover, as mentioned in the introduction, the complexity of $1$-LLL and $1$-PotLLL for input basis consisting of integer vectors is further bounded by $O(M^{n^2})$; our bound does not improve this bound either.

We emphasize that the main advantage of our result is its generality compared to the previous results.
Namely, in contrast to our result applicable to all of the four algorithms in a unified way, the proof of the bound $O(A^{n^3}\log M)$ in \cite{akh: 1lll} is specific to the case of $1$-LLL.
On the other hand, the bound $O(M^{n^2})$ for the case of integer inputs in $1$-LLL and $1$-PotLLL mentioned above is deduced from the fact that now $\Pot(\matfont{B})$ takes integer values and is monotonically decreasing during those algorithms.
The same argument for $1$-$S^2$LLL by using $\mss(\matfont{B})$ instead of $\Pot(\matfont{B})$ does not work due to the difference that $\mss(\matfont{B})$ is not necessarily an integer even for integer basis $\matfont{B}$.
For the case of $1$-DeepLLL, if the value $\Pot(\matfont{B})$ were also monotonically decreasing during the algorithm, then the same argument would yield a simple bound better than ours.
However, in fact $\Pot(\vecfont{B})$ is not monotonically decreasing in $1$-DeepLLL, as shown in the following example, therefore the previous argument might not be straightforwardly applicable to $1$-DeepLLL.

\begin{example}
We give examples showing that the value $\Pot(\matfont{B})$ (for integer inputs) is not monotonically decreasing in general for $1$-DeepLLL and $1$-$S^2$LLL.
We set $n = 3$ and consider the following size-reduced basis
\[
\matfont{B} =
\begin{pmatrix}
0 & -3 & 2 \\
3 & -2 & -2 \\
-2 & 0 & -2
\end{pmatrix} \enspace.
\]
By applying GSO to $\matfont{B}$, we obtain
\[
\matfont{B}^\ast =
\begin{pmatrix}
0 & -3 & 8/7 \\
3 & -8/13 & -12/7 \\
-2 & -12/13 & -18/7
\end{pmatrix} \enspace.
\]
Hence we have
\begin{align*}
\Pot(\matfont{B}) &= 
\norm{\begin{pmatrix} 0 \\ 3\\ -2 \end{pmatrix}}^{2\cdot3}
\norm{\begin{pmatrix} -3 \\ -8/13\\ -12/13 \end{pmatrix}}^{2\cdot2}
\norm{\begin{pmatrix} 8/7 \\ -12/7\\ -18/7 \end{pmatrix}}^{2\cdot1} \\
&= 13^3 \cdot \left( \frac{ 133 }{ 13 } \right)^2 \cdot \frac{ 76 }{ 7 }
= 13 \cdot 7 \cdot 19^2 \cdot 76
= 2,496,676 \enspace.
\end{align*}
On the other hand, when $\matfont{B}$ is input to $1$-DeepLLL, $k = 2$ is incremented to $k = 3$ in the first \texttt{while} block (as $\| \vecfont{b}_2 \| = \| \vecfont{b}_1 \|$), and then $\mathrm{SWAP}_{1,3}$ is executed in the second \texttt{while} block (as $\| \vecfont{b}_3 \|^2 = 12 < 13 = \| \vecfont{b}_1 \|^2$).
The result is
\[
\sigma_{1,3}(\matfont{B}) =
\begin{pmatrix}
2 & 0 & -3 \\
-2 & 3 & -2 \\
-2 & -2 & 0
\end{pmatrix} \enspace,
\]
which is size-reduced.
By applying GSO, we obtain
\[
(\sigma_{1,3}(\matfont{B}))^\ast =
\begin{pmatrix}
2 & 1/3 & -5/2 \\
-2 & 8/3 & -1 \\
-2 & -7/3 & -3/2
\end{pmatrix} \enspace.
\]
Hence we have
\begin{align*}
\Pot(\sigma_{1,3}(\matfont{B})) &=
\norm{\begin{pmatrix} 2 \\ -2\\ -2 \end{pmatrix}}^{2\cdot3}
\norm{\begin{pmatrix} 1/3 \\ 8/3\\ -7/3 \end{pmatrix}}^{2\cdot2}
\norm{\begin{pmatrix} -5/2 \\ -1\\ -3/2 \end{pmatrix}}^{2\cdot1} \\
&= 12^3 \cdot \left( \frac{ 38 }{ 3 } \right)^2 \cdot \frac{ 19 }{ 2 }
= 3 \cdot 4^3 \cdot 2 \cdot 19^2 \cdot 19
= 2,633,856 > \Pot(\matfont{B}) \enspace,
\end{align*}
which shows that the value $\Pot(\matfont{B})$ is not monotonically decreasing in $1$-DeepLLL.

Similarly, for $1$-$S^2$LLL, we set $n = 3$ and consider the following size-reduced basis
\[
\matfont{B} =
\begin{pmatrix}
3 & 1 & 1 \\
1 & -1 & -2 \\
-1 & 2 & -2
\end{pmatrix} \enspace.
\]
By applying GSO to $\matfont{B}$, we obtain
\[
\matfont{B}^\ast = \begin{pmatrix}
3 & 1 & 23/66 \\
1 & -1 & -161/66 \\
-1 & 2 & -46/33
\end{pmatrix} \enspace.
\]
Hence we have
\begin{align*}
\Pot(\matfont{B}) &= 
\norm{\begin{pmatrix} 3 \\ 1\\ -1 \end{pmatrix}}^{2\cdot3}
\norm{\begin{pmatrix} 1 \\ -1\\ 2 \end{pmatrix}}^{2\cdot2}
\norm{\begin{pmatrix} 23/66 \\ -161/66\\ -46/33 \end{pmatrix}}^{2\cdot1} \\
&= 11^3 \cdot 6^2 \cdot \frac{ 529 }{ 66 }
= 11^2 \cdot 6 \cdot 529
= 384,054 \enspace.
\end{align*}
On the other hand, when $\matfont{B}$ is input to $1$-$S^2$LLL, $k = 2$ is incremented to $k = 3$ in the first \texttt{while} block (as now $\langle \vecfont{b}_1, \vecfont{b}_2 \rangle = 0$ and $\mathrm{SWAP}_{1,2}$ does not change the values of $B_i$'s, therefore $S_{12} = \mss(\matfont{B}) - \mss(\sigma_{1,2}(\matfont{B})) = 0$).
For the next \texttt{while} block, we have
\[
\sigma_{1,3}(\matfont{B}) = \begin{pmatrix}
1 & 3 & 1 \\
-2 & 1 & -1 \\
-2 & -1 & 2
\end{pmatrix}
\,,\,
\sigma_{2,3}(\matfont{B}) = \begin{pmatrix}
3 & 1 & 1 \\ 
1 & -2 & -1 \\
-1 & -2 & 2
\end{pmatrix} \enspace,
\]
which are size-reduced, and
\[ 
(\sigma_{1,3}(\matfont{B}))^\ast =
\begin{pmatrix}
1 & 8/3 & 46/45 \\
-2 & 5/3 & -23/18 \\
-2 & -1/3 & 161/90
\end{pmatrix}
\,,\,
(\sigma_{2,3}(\matfont{B}))^\ast =
\begin{pmatrix}
3 & 2/11 & 46/45 \\
1 & -25/11 & -23/18 \\
-1 & -19/11 & 161/90
\end{pmatrix} \enspace.
\]
Then we have
\[
\mss(\sigma_{1,3}(\matfont{B})) = \norm{\begin{pmatrix} 1 \\ -2 \\ -2 \end{pmatrix}}^2 + \norm{\begin{pmatrix} 8/3 \\ 5/3 \\ -1/3 \end{pmatrix}}^2 + \norm{\begin{pmatrix} 46/45 \\ -23/18 \\ 161/90 \end{pmatrix}}^2 = \frac{ 2239 }{ 90 } \enspace,
\]
\[
\mss(\sigma_{2,3}(\matfont{B})) = \norm{\begin{pmatrix} 3 \\ 1 \\ -1 \end{pmatrix}}^2 + \norm{\begin{pmatrix} 2/11 \\ -25/11 \\ -19/11 \end{pmatrix}}^2 + \norm{\begin{pmatrix} 46/45 \\ -23/18 \\ 161/90 \end{pmatrix}}^2 = \frac{ 24809 }{ 990 } \enspace,
\]
and $S_{13} = \mss(\matfont{B}) - \mss(\sigma_{1,3}(\matfont{B})) > \mss(\matfont{B}) - \mss(\sigma_{2,3}(\matfont{B})) = S_{23}$.
Moreover, we have $S_{13} = \mss(\matfont{B}) - \mss(\sigma_{1,3}(\matfont{B})) = 68/495 > 0$, therefore $\mathrm{SWAP}_{1,3}$ is executed in the \texttt{while} block.
Now we have
\begin{align*}
\Pot(\sigma_{1,3}(\matfont{B})) &=
\norm{\begin{pmatrix} 1 \\ -2\\ -2 \end{pmatrix}}^{2\cdot3}
\norm{\begin{pmatrix} 8/3 \\ 5/3\\ -1/3 \end{pmatrix}}^{2\cdot2}
\norm{\begin{pmatrix} 46/45 \\ -23/18\\ 161/90 \end{pmatrix}}^{2\cdot1} \\
&= 9^3 \cdot 10^2 \cdot \frac{ 529 }{ 90 }
= 9^2 \cdot 10 \cdot 529
= 428,490 > \Pot(\matfont{B}) \enspace,
\end{align*}
which shows that the value $\Pot(\matfont{B})$ is not monotonically decreasing in $1$-$S^2$LLL.
\end{example}

\section{Experimental Comparison of the Cases $\delta < 1$ and $\delta = 1$}
\label{sec:experiment}

In this section, we describe our computer experiments to compare practical behaviors of LLL algorithm and its variants for two cases $\delta = 1$ and $\delta < 1$ (more precisely, $\delta = 0.99$ which is one of the popular parameter choices in cryptographic applications). 

We implemented the four algorithms dealt with in Section \ref{sec:halt_property} (LLL, DeepLLL, PotLLL, $S^2$LLL) and also variants of DeepLLL algorithm where the \texttt{for} loop is executed only for the range $k - 5 \leq i \leq k - 1$ (denoted here by \lq\lq Deep-$5$\rq\rq) or $k - 10 \leq i \leq k - 1$ (denoted here by \lq\lq Deep-$10$\rq\rq).
The machine environment is: Windows WSL 2, Intel(R) Core(TM) i7-8650 CPU @ 1.90 GHz 2.11 GHz, 16 GB main memory.
Instead of the real execution times, we counted the numbers of exchanges $\mathrm{SWAP}_{*,*}$ of basis vectors, which is an implementation-independent quantity commonly available for any of these algorithms and is intuitively the dominant part among the whole algorithm.
The dimensions of the lattices were set to be $n = 10, 15, 20, 25, 30, 35, 40$.
We dealt with two types of input basis.
One is obtained by taking the projection of a problem instance of Darmstadt SVP Challenge \cite{SVPChallenge}, that is, given an input basis $\matfont{B}$ of dimension $40$ from the challenge (which is the minimum available dimension), its upper left $n \times n$ submatrix is used as our input basis, denoted here by $\matfont{B}_1$.
Note that the original $\matfont{B}$ is an upper triangular matrix, so is our input $\matfont{B}_1$.
Note also that an instance of SVP Challenge is generated by specifying a random seed; we used seeds $0, 1, 2, 3, 4$ to obtain five original bases $\matfont{B}$ above.
The other, denoted here by $\matfont{B}_2$, is obtained by randomly generating a lattice with $\vol(L) = 1$.
More precisely, $\matfont{B}_2$ is generated by starting from the identity matrix and iterating, $1000$ times in total, addition of $\pm 1$ times some column to some other column, where the two columns and the sign $\pm 1$ at each of the $1000$ steps are randomly selected.
Hence $\matfont{B}_2$ is an integer matrix.
We also used five inputs for each dimension, specified by seeds from $0$ to $4$.

Table \ref{tab:ratio_for_different_delta} shows the maximum values, among the five inputs for each algorithm, dimension, and input type, of the ratios of numbers of basis exchanges for the case $\delta = 1$ relative to the case $\delta = 0.99$ with the same input; \lq\lq N/A\rq\rq{} means that the algorithm did not halt within our experiment time.
The exact results of the experiments are shown in Tables \ref{tbl:matrix-data_dim10-15}, \ref{tbl:matrix-data_dim20-25}, \ref{tbl:matrix-data_dim30-35}, and \ref{tbl:matrix-data_dim40} at the appendix.
For the five algorithms other than $S^2$LLL, the ratios calculated in our experiments are fairly small (smaller than $3$, and mostly close to $1$) and seem to be not drastically growing as $n$ increases.
For the case of $S^2$LLL, the ratio grows more rapidly especially for the input type $\matfont{B}_2$, but the growth seems to be linear or at most quadratic in $n$.
These results suggest that, contrary to the significant difference of our hyperexponentially growing upper bound for the case $\delta = 1$ from the known polynomial upper bound (except for DeepLLL) for the case $\delta < 1$, the computational complexity for the two cases $\delta = 1$ and $\delta < 1$ in practice may be closer and accordingly, there might be hope to improve the upper bound for the case $\delta = 1$ much further.

\begin{table}[t!]
\centering
\caption{Maximum ratios of numbers of basis exchanges for the case $\delta = 1$ relative to the case $\delta = 0.99$ with the same input; here \lq\lq N/A\rq\rq{} means that the algorithm did not halt within our experiment time}
\label{tab:ratio_for_different_delta}
\begin{tabular}{cc|ccccccc}
& & \multicolumn{7}{c}{Dimension $n$} \\ \cline{3-9}
Algorithm & Input & $10$ & $15$ & $20$ & $25$ & $30$ & $35$ & $40$ \\ \hline
\multirow{2}{*}{LLL}        & $\matfont{B}_1$ & $1.025$ & $1.037$ & $1.066$ & $1.077$ & $1.106$ & $1.102$ & $1.097$ \\
                                 & $\matfont{B}_2$ & $1.035$ & $1.050$ & $1.083$ & $1.115$ & $1.065$ & $1.114$ & $1.100$ \\ \hline
\multirow{2}{*}{Deep}      & $\matfont{B}_1$ & $1.043$ & $1.096$ & $1.222$ & $1.159$ & $1.237$ & N/A & N/A \\
                                 & $\matfont{B}_2$ & $1.077$ & $1.131$ & $1.241$ & $1.331$ & $1.450$ & $1.473$ & N/A \\ \hline
\multirow{2}{*}{Deep-5}   & $\matfont{B}_1$ & $1.055$ & $1.101$ & $1.126$ & $1.126$ & $1.126$ & $1.126$ & $1.126$ \\
                                 & $\matfont{B}_2$ & $1.069$ & $1.299$ & $1.223$ & $2.193$ & $1.385$ & $2.544$ & $1.144$ \\ \hline
\multirow{2}{*}{Deep-10} & $\matfont{B}_1$ & $1.043$ & $1.088$ & $1.095$ & $1.299$ & $1.299$ & $1.299$ & $1.299$ \\
                                 & $\matfont{B}_2$ & $1.077$ & $1.447$ & $1.564$ & $2.033$ & $1.689$ & $1.169$ & $1.498$ \\ \hline
\multirow{2}{*}{Pot}        & $\matfont{B}_1$ & $1.042$ & $1.038$ & $1.063$ & $1.070$ & $1.084$ & $1.078$ & $1.098$ \\
                                 & $\matfont{B}_2$ & $1.041$ & $1.042$ & $1.080$ & $1.133$ & $1.058$ & $1.080$ & $1.093$ \\ \hline
\multirow{2}{*}{$S^2$}     & $\matfont{B}_1$ & $1.604$ & $2.346$ & $3.256$ & $4.263$ & $5.373$ & $6.479$ & $7.629$ \\
                                 & $\matfont{B}_2$ & $2.075$ & $3.375$ & $4.976$ & $6.872$ & $8.882$ & $11.381$ & $13.228$ \\ \hline
\end{tabular}
\end{table}

\paragraph{Acknowledgements.}

This research was partially supported by the Ministry of Internal Affairs and
Communications SCOPE Grant Number 182103105 and by JSPS KAKENHI Grant Number JP19H01804.

\appendix

\section*{Appendix}

Here we include (as mentioned in Section \ref{sec:experiment}) the tables showing the detailed numbers of exchanges for basis vectors in LLL algorithm and its variants with two choices of parameters $\delta = 0.99$ and $\delta = 1$.
We note that for some inputs for Deep-5 and Deep-10, the number of basis exchanges for the case $\delta = 1$ is even smaller than that for the case $\delta = 0.99$ as opposed to the intuition.
This would be explained by noting that for DeepLLL algorithm and its variants (i.e., Deep-5 and Deep-10), the running time depends not only on the number of basis exchanges but also on the distribution of the index $i$ in \texttt{for} loop at which the basis exchange occurred, the latter being different in general for the two cases $\delta = 0.99$ and $\delta = 1$.

\begin{table}[p]
\centering
\caption{Numbers of exchanges for basis vectors in LLL and its variants with $\delta = 0.99$ (upper rows) and $\delta = 1$ (lower rows) for dimensions $n \in \{10,15\}$; here \lq\lq Deep\rq\rq, \lq\lq Pot\rq\rq, and \lq\lq $S^2$\rq\rq{} stand for DeepLLL, PotLLL, and $S^2$LLL, respectively}
\label{tbl:matrix-data_dim10-15}
\begin{tabular}{ccc|c|rrrrrr}
 & & & & \multicolumn{6}{c}{\# of Exchanges} \\
$n$ & Input & Seed & $M$ & LLL & Deep & Deep-5 & Deep-10 & Pot & $S^2$ \\ \hline
\multirow{20}{*}{10} & \multirow{10}{*}{$\matfont{B}_1$} & \multirow{2}{*}{0} & \multirow{2}{*}{$2.12\times10^{120}$} & 1862 & 717 & 807 & 717 & 706 & 484 \\
& & & & 1901 & 742 & 851 & 742 & 726 & 720 \\ \cline{3-10}
 & & \multirow{2}{*}{1} & \multirow{2}{*}{$1.41\times10^{120}$} & 2018 & 778 & 853 & 778 & 766 & 485 \\
& & & & 2046 & 801 & 865 & 801 & 785 & 778 \\ \cline{3-10}
&  & \multirow{2}{*}{2} & \multirow{2}{*}{$1.86\times10^{120}$} & 2036 & 771 & 854 & 771 & 743 & 488 \\
& & & & 2066 & 790 & 882 & 790 & 761 & 748 \\ \cline{3-10}
 & & \multirow{2}{*}{3} & \multirow{2}{*}{$1.89\times10^{120}$} & 2052 & 774 & 849 & 774 & 757 & 480 \\
& & & & 2087 & 787 & 858 & 787 & 767 & 760 \\ \cline{3-10}
 & & \multirow{2}{*}{4} & \multirow{2}{*}{$2.07\times10^{120}$} & 1998 & 767 & 870 & 767 & 742 & 482 \\
& & & & 2047 & 800 & 905 & 800 & 773 & 762 \\ \cline{2-10}
& \multirow{10}{*}{$\matfont{B}_2$} & \multirow{2}{*}{0} & \multirow{2}{*}{$3.46\times10^{18}$} & 1508 & 486 & 579 & 486 & 485 & 240 \\
 & & & & 1556 & 517 & 601 & 517 & 504 & 498 \\ \cline{3-10}
 &  & \multirow{2}{*}{1} & \multirow{2}{*}{$1.10\times10^{17}$} & 1520 & 507 & 589 & 507 & 502 & 251 \\
 & & & & 1559 & 540 & 606 & 540 & 513 & 499 \\ \cline{3-10}
 & & \multirow{2}{*}{2} & \multirow{2}{*}{$1.33\times10^{16}$} & 1357 & 457 & 522 & 457 & 419 & 215 \\
 & & & & 1383 & 470 & 558 & 470 & 436 & 428 \\ \cline{3-10}
 & & \multirow{2}{*}{3} & \multirow{2}{*}{$8.59\times10^{15}$} & 1325 & 432 & 517 & 432 & 428 & 237 \\
 & & & & 1371 & 451 & 529 & 451 & 433 & 432 \\ \cline{3-10}
 & & \multirow{2}{*}{4} & \multirow{2}{*}{$2.16\times10^{18}$} & 1583 & 521 & 610 & 521 & 508 & 253 \\
 & & & & 1627 & 561 & 641 & 561 & 516 & 514 \\ \hline
%%%%
\multirow{20}{*}{15} & \multirow{10}{*}{$\matfont{B}_1$} & \multirow{2}{*}{0} & \multirow{2}{*}{$2.12\times10^{120}$} & 4465 & 1774 & 1337 & 1805 & 1570 & 710  \\
 & & & & 4623 & 1883 & 1384 & 1964 & 1611 & 1598 \\ \cline{3-10}
 & & \multirow{2}{*}{1} & \multirow{2}{*}{$1.41\times10^{120}$} & 4460 & 1691 & 1145 & 1514 & 1555 & 696  \\
 & & & & 4600 & 1853 & 1163 & 1607 & 1612 & 1564 \\ \cline{3-10}
 & & \multirow{2}{*}{2} & \multirow{2}{*}{$1.86\times10^{120}$} & 4489 & 1765 & 1338 & 1839 & 1556 & 686  \\
 & & & & 4553 & 1832 & 1185 & 1860 & 1597 & 1539 \\ \cline{3-10}
 & & \multirow{2}{*}{3} & \multirow{2}{*}{$1.89\times10^{120}$} & 4602 & 1888 & 1314 & 1495 & 1608 & 696  \\
 & & & & 4748 & 1998 & 1338 & 1379 & 1650 & 1633 \\ \cline{3-10}
 & & \multirow{2}{*}{4} & \multirow{2}{*}{$2.07\times10^{120}$} & 4624 & 1891 & 1340 & 1598 & 1600 & 704  \\
 & & & & 4794 & 2066 & 1475 & 1712 & 1661 & 1588 \\ \cline{2-10}
& \multirow{10}{*}{$\matfont{B}_2$} & \multirow{2}{*}{0} & \multirow{2}{*}{$3.04\times10^{13}$} & 3241 & 1270 & 471 & 1241 & 1062 & 341  \\
 & & & & 3402 & 1344 & 612 & 1307 & 1107 & 1060 \\ \cline{3-10}
 & & \multirow{2}{*}{1} & \multirow{2}{*}{$1.03\times10^{12}$} & 3385 & 1428 & 674 & 1304 & 1162 & 345  \\
 & & & & 3473 & 1520 & 518 & 1441 & 1176 & 1154 \\ \cline{3-10}
 & & \multirow{2}{*}{2} & \multirow{2}{*}{$1.59\times10^{13}$} & 3436 & 1322 & 547 & 1305 & 1120 & 328  \\
 & & & & 3609 & 1495 & 543 & 1386 & 1166 & 1107 \\ \cline{3-10}
 & & \multirow{2}{*}{3} & \multirow{2}{*}{$2.58\times10^{12}$} & 3786 & 1473 & 687 & 1487 & 1257 & 360  \\
 & & & & 3948 & 1605 & 735 & 1239 & 1299 & 1198 \\ \cline{3-10}
 & & \multirow{2}{*}{4} & \multirow{2}{*}{$2.58\times10^{12}$} & 3195 & 1246 & 688 & 804  & 1027 & 312  \\
 & & & & 3282 & 1272 & 622 & 1163 & 1061 & 1029 \\ \hline
\end{tabular}
\end{table}
%%%%%%%%%%%%%%%%%%%
\begin{table}[p]
\centering
\caption{Numbers of exchanges for basis vectors in LLL and its variants with $\delta = 0.99$ (upper rows) and $\delta = 1$ (lower rows) for dimensions $n \in \{20,25\}$; here \lq\lq Deep\rq\rq, \lq\lq Pot\rq\rq, and \lq\lq $S^2$\rq\rq{} stand for DeepLLL, PotLLL, and $S^2$LLL, respectively}
\label{tbl:matrix-data_dim20-25}
\begin{tabular}{ccc|c|rrrrrr}
 & & & & \multicolumn{6}{c}{\# of Exchanges} \\
$n$ & Input & Seed & $M$ & LLL & Deep & Deep-5 & Deep-10 & Pot & $S^2$ \\ \hline
\multirow{20}{*}{20} & \multirow{10}{*}{$\matfont{B}_1$} & \multirow{2}{*}{0} & \multirow{2}{*}{$2.12\times10^{120}$} & 7995 & 3788 & 1345 & 2287 & 2860 & 903  \\
 & & & & 8149 & 4202 & 1515 & 2414 & 3009 & 2830 \\ \cline{3-10}
 & & \multirow{2}{*}{1} & \multirow{2}{*}{$1.41\times10^{120}$} & 7952 & 3736 & 1371 & 2147 & 2868 & 889  \\
 & & & & 8476 & 4564 & 1352 & 2352 & 3032 & 2841 \\ \cline{3-10}
 & & \multirow{2}{*}{2} & \multirow{2}{*}{$1.86\times10^{120}$} & 7948 & 4066 & 1367 & 2429 & 2896 & 885  \\
 & & & & 8310 & 4597 & 1185 & 1942 & 3049 & 2831 \\ \cline{3-10}
 & & \multirow{2}{*}{3} & \multirow{2}{*}{$1.89\times10^{120}$} & 7892 & 3724 & 1314 & 2433 & 2826 & 871  \\
 & & & & 8223 & 4146 & 1338 & 2469 & 2948 & 2836 \\ \cline{3-10}
 & & \multirow{2}{*}{4} & \multirow{2}{*}{$2.07\times10^{120}$} & 8188 & 4437 & 1340 & 2567 & 2976 & 892  \\
 & & & & 8556 & 4841 & 1475 & 2258 & 3164 & 2903 \\ \cline{2-10}
& \multirow{10}{*}{$\matfont{B}_2$} & \multirow{2}{*}{0} & \multirow{2}{*}{$3.31\times10^{9}$} & 5869 & 3245 & 428 & 956  & 2095 & 423  \\
 & & & & 6356 & 4027 & 454 & 1495 & 2226 & 2044 \\ \cline{3-10}
 & & \multirow{2}{*}{1} & \multirow{2}{*}{$4.56\times10^{10}$} & 6178 & 3490 & 387 & 1435 & 2217 & 446  \\
 & & & & 6462 & 3959 & 408 & 1464 & 2300 & 2203 \\ \cline{3-10}
 & & \multirow{2}{*}{2} & \multirow{2}{*}{$2.32\times10^{10}$} & 6454 & 3659 & 452 & 1323 & 2361 & 452  \\
 & & & & 6831 & 4068 & 553 & 1588 & 2435 & 2249 \\ \cline{3-10}
 & & \multirow{2}{*}{3} & \multirow{2}{*}{$1.85\times10^{8}$}  & 5616 & 3127 & 377 & 1337 & 2001 & 396  \\
 & & & & 5948 & 3373 & 422 & 1250 & 2125 & 1933 \\ \cline{3-10}
 & & \multirow{2}{*}{4} & \multirow{2}{*}{$1.61\times10^{10}$} & 5956 & 3028 & 566 & 1219 & 1994 & 434  \\
 & & & & 6348 & 3296 & 489 & 1506 & 2153 & 2011 \\ \hline
%%%%
\multirow{20}{*}{25} & \multirow{10}{*}{$\matfont{B}_1$} & \multirow{2}{*}{0} & \multirow{2}{*}{$2.12\times10^{120}$} & 12183 & 7995 & 1345 & 2623 & 4605 & 1077 \\
& & & & 12521 & 9036 & 1515 & 2414 & 4872 & 4591 \\ \cline{3-10}
 & & \multirow{2}{*}{1} & \multirow{2}{*}{$1.41\times10^{120}$} & 12140 & 8135 & 1371 & 2147 & 4574 & 1061 \\
& & & & 13078 & 9430 & 1352 & 2789 & 4881 & 4476 \\ \cline{3-10}
 & & \multirow{2}{*}{2} & \multirow{2}{*}{$1.86\times10^{120}$} & 11828 & 7673 & 1367 & 2429 & 4394 & 1040 \\
& & & & 12626 & 8587 & 1185 & 1942 & 4700 & 4215 \\ \cline{3-10}
 & & \multirow{2}{*}{3} & \multirow{2}{*}{$1.89\times10^{120}$} & 11950 & 7738 & 1314 & 2524 & 4405 & 1046 \\
& & & & 12471 & 8492 & 1338 & 2469 & 4556 & 4366 \\ \cline{3-10}
 & & \multirow{2}{*}{4} & \multirow{2}{*}{$2.07\times10^{120}$} & 12074 & 8587 & 1340 & 2567 & 4555 & 1072 \\
& & & & 12826 & 9731 & 1475 & 2258 & 4876 & 4535 \\ \cline{2-10}
& \multirow{10}{*}{$\matfont{B}_2$} & \multirow{2}{*}{0} & \multirow{2}{*}{$9.54\times10^{7}$} & 8686 & 6555 & 388 & 666  & 3294 & 494  \\
 & & & & 9174 & 7998 & 213 & 705  & 3460 & 3264 \\ \cline{3-10}
 & & \multirow{2}{*}{1} & \multirow{2}{*}{$1.28\times10^{8}$} & 8527 & 6922 & 311 & 1085 & 3444 & 483  \\
 & & & & 9508 & 7813 & 315 & 764  & 3570 & 3319 \\ \cline{3-10}
 & & \multirow{2}{*}{2} & \multirow{2}{*}{$4.32\times10^{8}$} & 8802 & 6693 & 176 & 963  & 3218 & 484  \\
 & & & & 9541 & 8259 & 386 & 1191 & 3470 & 3178 \\ \cline{3-10}
 & & \multirow{2}{*}{3} & \multirow{2}{*}{$6.48\times10^{7}$} & 8779 & 6518 & 318 & 861  & 3362 & 477  \\
 & & & & 8942 & 8674 & 326 & 850  & 3525 & 3235 \\ \cline{3-10}
 & & \multirow{2}{*}{4} & \multirow{2}{*}{$2.23\times10^{8}$} & 8610 & 7051 & 269 & 607  & 3252 & 480  \\
 & & & & 9329 & 8608 & 240 & 1234 & 3686 & 3298 \\ \hline
\end{tabular}
\end{table}
%%%%%%%%%%%%%%%%%%%
\begin{table}[p]
\centering
\caption{Numbers of exchanges for basis vectors in LLL and its variants with $\delta = 0.99$ (upper rows) and $\delta = 1$ (lower rows) for dimensions $n \in \{30,35\}$; here \lq\lq Deep\rq\rq, \lq\lq Pot\rq\rq, and \lq\lq $S^2$\rq\rq{} stand for DeepLLL, PotLLL, and $S^2$LLL, respectively; and \lq\lq N/A\rq\rq{} means that the algorithm did not halt within our experiment time}
\label{tbl:matrix-data_dim30-35}
\begin{tabular}{ccc|c|rrrrrr}
 & & & & \multicolumn{6}{c}{\# of Exchanges} \\
$n$ & Input & Seed & $M$ & LLL & Deep & Deep-5 & Deep-10 & Pot & $S^2$ \\ \hline
\multirow{20}{*}{30} & \multirow{10}{*}{$\matfont{B}_1$} & \multirow{2}{*}{0} & \multirow{2}{*}{$2.12\times10^{120}$} & 16649 & 14587 & 1345 & 2623 & 6640 & 1218 \\
& & & & 17321 & 17425 & 1515 & 2414 & 6933 & 6439 \\ \cline{3-10}
 & & \multirow{2}{*}{1} & \multirow{2}{*}{$1.41\times10^{120}$} & 16392 & 14938 & 1371 & 2147 & 6751 & 1212 \\
& & & & 18124 & 18261 & 1352 & 2789 & 7138 & 6512 \\ \cline{3-10}
 & & \multirow{2}{*}{2} & \multirow{2}{*}{$1.86\times10^{120}$} & 16152 & 14273 & 1367 & 2429 & 6473 & 1183 \\
& & & & 17663 & 17661 & 1185 & 1942 & 6977 & 6194 \\ \cline{3-10}
 & & \multirow{2}{*}{3} & \multirow{2}{*}{$1.89\times10^{120}$} & 16436 & 14535 & 1314 & 2524 & 6418 & 1184 \\
& & & & 17332 & 17865 & 1338 & 2469 & 6772 & 6284 \\ \cline{3-10}
 & & \multirow{2}{*}{4} & \multirow{2}{*}{$2.07\times10^{120}$} & 16313 & 14898 & 1340 & 2567 & 6570 & 1226 \\
& & & & 17679 & 17482 & 1475 & 2258 & 7124 & 6398 \\ \cline{2-10}
& \multirow{10}{*}{$\matfont{B}_2$} & \multirow{2}{*}{0} & \multirow{2}{*}{$9.53\times10^{6}$} & 11442 & 11944 & 177 & 514 & 4727 & 534  \\
 & & & & 12184 & 15218 & 171 & 868 & 4999 & 4540 \\ \cline{3-10}
 & & \multirow{2}{*}{1} & \multirow{2}{*}{$1.13\times10^{7}$} & 11268 & 12261 & 286 & 593 & 4756 & 503  \\
 & & & & 11868 & 15072 & 147 & 828 & 4897 & 4231 \\ \cline{3-10}
 & & \multirow{2}{*}{2} & \multirow{2}{*}{$1.26\times10^{7}$} & 11930 & 12651 & 256 & 820 & 5046 & 544  \\
 & & & & 12560 & 15470 & 217 & 568 & 5324 & 4832 \\ \cline{3-10}
 & & \multirow{2}{*}{3} & \multirow{2}{*}{$1.26\times10^{7}$} & 11938 & 11441 & 208 & 592 & 4877 & 549  \\
 & & & & 12236 & 16589 & 288 & 817 & 5143 & 4610 \\ \cline{3-10}
 & & \multirow{2}{*}{4} & \multirow{2}{*}{$3.47\times10^{6}$} & 11371 & 12579 & 199 & 662 & 4996 & 525  \\
 & & & & 12012 & 15400 & 196 & 780 & 5118 & 4638 \\ \hline
%%%%
\multirow{20}{*}{35} & \multirow{10}{*}{$\matfont{B}_1$} & \multirow{2}{*}{0} & \multirow{2}{*}{$2.12\times10^{120}$} & 21221 & N/A & 1345 & 2623 & 8897 & 1340 \\
 & & & & 22548 & N/A & 1515 & 2414 & 9487 & 8569 \\ \cline{3-10}
 & & \multirow{2}{*}{1} & \multirow{2}{*}{$1.41\times10^{120}$} & 21432 & N/A & 1371 & 2147 & 9191 & 1333 \\
 & & & & 23010 & N/A & 1352 & 2789 & 9695 & 8548 \\ \cline{3-10}
 & & \multirow{2}{*}{2} & \multirow{2}{*}{$1.86\times10^{120}$} & 20835 & N/A & 1367 & 2429 & 8797 & 1316 \\
 & & & & 22882 & N/A & 1185 & 1942 & 9486 & 8274 \\ \cline{3-10}
 & & \multirow{2}{*}{3} & \multirow{2}{*}{$1.89\times10^{120}$} & 21096 & N/A & 1314 & 2524 & 8741 & 1314 \\
 & & & & 22332 & N/A & 1338 & 2469 & 9213 & 8513 \\ \cline{3-10}
 & & \multirow{2}{*}{4} & \multirow{2}{*}{$2.07\times10^{120}$} & 20816 & N/A & 1340 & 2567 & 9031 & 1353 \\
 & & & & 22940 & N/A & 1475 & 2258 & 9623 & 8593 \\ \cline{2-10}
& \multirow{10}{*}{$\matfont{B}_2$} & \multirow{2}{*}{0} & \multirow{2}{*}{$6.70\times10^{6}$} & 15685 & 24239 & 135 & 426 & 6954 & 582  \\
 & & & & 16961 & 33381 & 214 & 498 & 7512 & 6624 \\ \cline{3-10}
 & & \multirow{2}{*}{1} & \multirow{2}{*}{$1.38\times10^{6}$} & 14205 & 22975 & 131 & 554 & 6550 & 560  \\
 & & & & 15825 & 31369 & 195 & 567 & 6863 & 5896 \\ \cline{3-10}
 & & \multirow{2}{*}{2} & \multirow{2}{*}{$1.81\times10^{6}$} & 13789 & 21574 & 165 & 612 & 6018 & 530  \\
 & & & & 15055 & 25067 & 179 & 442 & 6480 & 5888 \\ \cline{3-10}
 & & \multirow{2}{*}{3} & \multirow{2}{*}{$1.88\times10^{6}$} & 15201 & 21214 & 206 & 489 & 6652 & 550  \\
 & & & & 15393 & 31244 & 188 & 551 & 6984 & 6259 \\ \cline{3-10}
 & & \multirow{2}{*}{4} & \multirow{2}{*}{$1.38\times10^{6}$} & 14941 & 22677 & 136 & 545 & 6610 & 579  \\
 & & & & 15645 & 29246 & 346 & 614 & 6899 & 6299 \\ \hline
\end{tabular}
\end{table}
%%%%%%%%%%%%%%%%%%%
\begin{table}[p]
\centering
\caption{Numbers of exchanges for basis vectors in LLL and its variants with $\delta = 0.99$ (upper rows) and $\delta = 1$ (lower rows) for dimension $n = 40$; here \lq\lq Deep\rq\rq, \lq\lq Pot\rq\rq, and \lq\lq $S^2$\rq\rq{} stand for DeepLLL, PotLLL, and $S^2$LLL, respectively; and \lq\lq N/A\rq\rq{} means that the algorithm did not halt within our experiment time}
\label{tbl:matrix-data_dim40}
\begin{tabular}{ccc|c|rrrrrr}
 & & & & \multicolumn{6}{c}{\# of Exchanges} \\
$n$ & Input & Seed & $M$ & LLL & Deep & Deep-5 & Deep-10 & Pot & $S^2$ \\ \hline
\multirow{20}{*}{40} & \multirow{10}{*}{$\matfont{B}_1$} & \multirow{2}{*}{0} & \multirow{2}{*}{$2.12\times10^{120}$} & 26073 & N/A & 1345 & 2623 & 11530 & 1449  \\
 & & & & 27687 & N/A & 1515 & 2414 & 12232 & 11010 \\ \cline{3-10}
 & & \multirow{2}{*}{1} & \multirow{2}{*}{$1.41\times10^{120}$} & 26116 & N/A & 1371 & 2147 & 11690 & 1447  \\
 & & & & 27728 & N/A & 1352 & 2789 & 12246 & 10923 \\ \cline{3-10}
 & & \multirow{2}{*}{2} & \multirow{2}{*}{$1.86\times10^{120}$} & 25890 & N/A & 1367 & 2429 & 11119 & 1421  \\
 & & & & 28086 & N/A & 1185 & 1942 & 12213 & 10471 \\ \cline{3-10}
 & & \multirow{2}{*}{3} & \multirow{2}{*}{$1.89\times10^{120}$} & 25780 & N/A & 1314 & 2524 & 11321 & 1419  \\
 & & & & 27344 & N/A & 1338 & 2469 & 11964 & 10825 \\ \cline{3-10}
 & & \multirow{2}{*}{4} & \multirow{2}{*}{$2.07\times10^{120}$} & 25525 & N/A & 1340 & 2567 & 11608 & 1465  \\
 & & & & 28013 & N/A & 1475 & 2258 & 12277 & 10940 \\ \cline{2-10}
& \multirow{10}{*}{$\matfont{B}_2$} & \multirow{2}{*}{0} & \multirow{2}{*}{$1.32\times10^{5}$} & 17063 & N/A & 195 & 387 & 7869 & 572  \\
 & & & & 18537 & N/A & 223 & 407 & 8372 & 7513 \\ \cline{3-10}
 & & \multirow{2}{*}{1} & \multirow{2}{*}{$1.96\times10^{5}$} & 17715 & N/A & 146 & 317 & 7666 & 600  \\
 & & & & 18595 & N/A & 162 & 365 & 8031 & 7530 \\ \cline{3-10}
 & & \multirow{2}{*}{2} & \multirow{2}{*}{$2.70\times10^{5}$} & 17799 & N/A & 139 & 424 & 8542 & 610  \\
 & & & & 18434 & N/A & 138 & 492 & 8958 & 7731 \\ \cline{3-10}
 & & \multirow{2}{*}{3} & \multirow{2}{*}{$1.34\times10^{5}$} & 17886 & N/A & 188 & 297 & 8274 & 592  \\
 & & & & 19561 & N/A & 143 & 445 & 9040 & 7831 \\ \cline{3-10}
 & & \multirow{2}{*}{4} & \multirow{2}{*}{$2.10\times10^{5}$} & 15726 & N/A & 211 & 397 & 7592 & 540  \\
 & & & & 17298 & N/A & 217 & 450 & 7881 & 7028 \\ \hline
\end{tabular}
\end{table}

\begin{comment}

\begin{algorithm}[htbp]
  \caption{Potential LLL}
  \label{alg:potLLL}
  \SetKwComment{Comment}{$\vartriangleright$~}{}
  \SetCommentSty{textit}
  \DontPrintSemicolon
  
  \KwIn{Basis $B \in \Z^{n \times m}$, $\delta \in (1/4,1]$}
  \KwOut{A $\delta$-\plll reduced basis.}
  $\ell \leftarrow 1$\;
  \While{$\ell \leq n$}{
    	Size-reduce$(B)$\;\label{alg:potLLL:sizereduce}
	$k \leftarrow \argmin_{1\leq j\leq \ell}\pot(\sigma_{j,\ell}B)$\label{alg:potLLL:min}\;	
	\eIf{$\delta \cdot \pot(B)>\pot(\sigma_{k,\ell}B$)\label{alg:potLLL:if}}{
			$B \leftarrow \sigma_{k,\ell}B$\;
			$\ell \leftarrow k$ \;
	}{
		$\ell \leftarrow \ell+1$ \;
	}
  }
\Return $B$\;
\end{algorithm}

\end{comment}
\end{document}